\def\year{2020}\relax
\documentclass[letterpaper]{article} 
\usepackage{aaai20}  
\usepackage{times}  
\usepackage{helvet} 
\usepackage{courier}  
\usepackage[hyphens]{url}  
\usepackage{graphicx} 
\usepackage{graphicx}  
\usepackage{natbib}
\usepackage{amsmath,amsthm,amssymb,algorithm,algpseudocode,graphicx,tikz,caption,subcaption,xcolor}
\definecolor{light-gray}{gray}{0.7}

\algrenewcommand\algorithmicindent{0.43em}
\newtheorem{theorem}{Theorem}[section]
\newtheorem{proposition}[theorem]{Proposition}
\newtheorem{corollary}[theorem]{Corollary}

\algnotext{EndFor}
\algnotext{EndIf}
\algnotext{EndWhile}
\algnotext{EndProcedure}
\newcommand{\calN}{{{\cal N}}}

\theoremstyle{definition}
\newtheorem{definition}[theorem]{Definition}
\newtheorem{example}[theorem]{Example}

\allowdisplaybreaks

\begin{document}
%
\title{Individual-Based Stability in Hedonic Diversity Games
\thanks{A shortened version of this paper appears in the proceedings of AAAI-20}}
\author{Niclas Boehmer\thanks{
Most of this research was done when the first author was an MSc student 
at the University of Oxford}\\ 
TU Berlin \\ Berlin, Germany \\
niclas.boehmer@tu-berlin.de
\And    Edith Elkind \\ University of Oxford \\ Oxford, UK\\
elkind@cs.ox.ac.uk}
\maketitle
\begin{abstract}
In {\em hedonic diversity games (HDGs)}, recently introduced by \citet{BEI19}, 
each agent belongs to one of two classes (men and women, 
vegetarians and meat-eaters, junior and senior researchers), and agents'
preferences over coalitions are determined by the fraction of agents
from their class in each coalition. Bredereck et al.
~show that while an HDG may fail to have
a Nash stable (NS) or a core stable (CS) outcome, every HDG 
in which all agents have single-peaked preferences admits an individually stable (IS) outcome, 
which can be computed in polynomial time. In this work, 
we extend and strengthen these results in several ways. 
First, we establish that the problem of deciding if an HDG has an NS
outcome is NP-complete, but admits an XP algorithm with respect to the size of the smaller class. 
Second, we show that, in fact, all HDGs
admit IS outcomes that can be computed in polynomial time;
our algorithm for finding such outcomes is considerably simpler 
than that of Bredereck et al. We also consider two ways of generalizing 
the model of Bredereck et al.~to $k\ge 2$ classes. 
We complement our theoretical results by empirical analysis,  
comparing the IS outcomes found by our algorithm, the algorithm
of Bredereck et al.~and a natural better-response dynamics. 
\end{abstract}

\section{Introduction}
A number of French exchange students at a Spanish university have signed up for a game theory class.
All students taking the class are advised to form study groups to discuss the material and to work
on problem sheets. Now, some French students see this as an excellent opportunity to improve
their Spanish and would like to join study groups where no one else speaks French. Other students
have less confidence in their ability to communicate in Spanish and therefore want to be in a group
where at least a few other students speak French; in fact, some of the students prefer to be in
an exclusively French-speaking group. Spanish students, too, have different preferences over the fraction
of French students in their groups: some are eager to meet new friends, while others are worried
that language issues will affect their learning.

Many important aspects of the setting described in the previous paragraph can be captured
by the recently introduced framework of {\em hedonic diversity games (HDG)} \citep{BEI19}. In these 
games agents can be split into two classes (say, {\em red} and {\em blue}), and each 
agent has preferences over the fraction of red agents in their group. The outcome of a game
is a partition of agents into groups. This model is relevant 
for  analyzing a variety of application scenarios, ranging from interdisciplinary collaborations
to racial segregation. 

In their work, Bredereck et al.~aim to understand whether HDGs admit stable outcomes, 
for several common notions of stability for hedonic games, such as 
Nash stability, individual stability and core stability; the first two concepts are based
on deviations by individual agents, while the third concept captures resilience
against group deviations. Bredereck et al.~show that an HDG may fail to have a Nash 
stable outcome or a core stable outcome and that deciding if an HDG has a core stable
outcome is NP-complete. For individual stability, they get a positive result under the additional 
assumption that agents' preferences are single-peaked, i.e., each agent $i$ 
has a preferred ratio of red agents in her group (say, $\rho_i$), 
and for any two ratios $\rho, \rho'$ such that
$\rho<\rho'\le \rho_i$ or $\rho_i\le \rho'<\rho$ she prefers 
a group with ratio $\rho'$ to a group with ratio $\rho$.
Specifically, Bredereck et al.~show that every HDG with single-peaked preferences admits
an individually stable outcome and describe a polynomial-time algorithm for finding some 
such outcome. Their work leaves open the question whether HDGs with non-single-peaked 
preferences always have an individually stable outcome.

\smallskip

\noindent{\bf Our Contribution\ } In this paper, we answer two open questions from the paper
of Bredereck et al., as well as extend their model to an arbitrary 
number of agent classes. 

First (Section~\ref{sec:ns}), we show that deciding if an HDG has a Nash stable outcome 
is NP-complete. Our hardness result holds even if agents have dichotomous preferences, 
i.e., approve some ratios and disapprove the remaining ratios; in fact, it remains true
if each agent approves at most $4$ ratios. On the other hand,  we show that
the existence of a Nash stable outcome can be decided in polynomial time if the size
of one of the classes can be bounded by a constant. 

We then turn our attention 
to individual stability (Section~\ref{sec:is}). We describe a polynomial-time algorithm that finds
an individually stable outcome of any HDG; our algorithm is significantly simpler 
than that of Bredereck et al. However, we show that neither algorithm Pareto-dominates
the other: there are settings where the algorithm of Bredereck et al.~produces
a much better partition than our algorithm, and there are settings where the converse
is true.

In Section~\ref{sec:ktypes}, we consider two different ways of extending HDGs
to settings with $k$ agent classes, $k>2$. First, we consider a very general model, 
where each agent may have
arbitrary preferences over the ratios of different classes in her group. 
We show that our positive result for individual stability does not extend to this model:
we describe a game with $k=3$ that has no individually stable outcomes and prove that 
deciding the existence of such outcomes is NP-complete if $k\ge 5$.
We then propose a more restrictive model, where an agent only cares about the fraction 
of agents that belong to her class. We show that this model encompasses both HDGs
and another well-known class of hedonic games, namely, anonymous games, i.e., 
games where agents have preferences over the size of their group.

In Section~\ref{sec:emp}, we empirically compare the outcomes produced 
(1) by our algorithm for finding individually stable outcomes, 
(2) by the algorithm of Bredereck et al. and 
(3) by a natural better-response dynamics,
with respect to several measures, such as the average social welfare and the diversity 
of resulting groups.
We summarize our findings in Section~\ref{sec:concl}.
\smallskip

\noindent{\bf Related Work\ }
Hedonic games were introduced by \citet{Dreze1980} and have received a lot of attention
in the computational social choice literature, as they offer a simple, but
powerful formalism to study group formation in strategic settings; see, e.g., the survey
by \citet{Aziz2016}.
Hedonic diversity games share common features with two other well-studied classes of hedonic games, 
namely, anonymous games \citep{Bogomolnaia2002} and fractional hedonic games \citep{AzizBBHOP17}. 

In anonymous games, agents cannot distinguish among other agents, and therefore the only feature 
of a group that matters to them is its size. HDGs may appear to be more general than
anonymous games, since in HDGs each agent can distinguish between two classes of agents. 
Indeed, many proof techniques developed for anonymous games turn out to be relevant for HDGs. 
However, in a technical sense, anonymous games are not a subclass of HDGs, 
and some of the positive results for HDGs do not hold for anonymous games.
In particular, while we prove that every HDG has an individually stable outcome, 
\citet{Bogomolnaia2002}
show that this is not the case for anonymous games.
Throughout the paper, we compare our results for HDGs to relevant results 
for anonymous games, and in Section~\ref{sec:ktypes} we propose a succinct
representation formalism for hedonic games that captures 
both HDGs and anonymous games.

In fractional hedonic games, every agent assigns a numerical value to every other agent, 
and an agent's value for a group of size $s$ that includes her is equal to the sum of 
the values she assigns to the group members, divided by $s$. Now, if agents are divided
into two classes, so that each agent assigns the same value to all agents in each class, 
the resulting game is a hedonic diversity game. However, not all hedonic diversity games 
can be obtained in this fashion; in particular, an HDG where each agent prefers groups
that have the same number of agents from each class cannot be represented in this way.
Conversely, there are fractional hedonic games that cannot be represented as HDGs. 


\section{Preliminaries}\label{sec:prelim}
For every positive integer $n$, we write $[n]$ to denote the set $\{1, \dots, n\}$.

A {\em hedonic game} is a pair $G = (N, (\succeq_i)_{i\in N})$, 
where $N=[n]$ is the set of {\em agents}, 
and for each $i\in N$ the relation $\succeq_i$ is a weak order over all subsets of $N$
that contain $i$. 
The subsets of $N$ are called {\em coalitions}; the set of all
coalitions containing agent $i$ is denoted by $\calN(i)$. 
We refer to the set $N$ as the {\em grand coalition}. 

Given two coalitions $C, D\in\calN(i)$, 
we write $C\sim_i D$ if $C\succeq_i D$ and $D\succeq_i C$;
we write $C\succ_i D$ if $C\succeq_i D$ and $C\not\sim_i D$. 
We say that $i$ {\em weakly prefers}
$C$ to $D$ if $C\succeq_i D$; if $C\succ_i D$, we say that $i$
{\em strictly prefers} $C$ to $D$, and if $C\sim_i D$,
we say that $i$ is {\em indifferent} between $C$ and $D$. 
For succinctness, when describing agents' preferences, we often omit 
coalitions $C$ with $\{i\}\succ_i C$. 

An {\em outcome} of a hedonic game with the set of agents $N$ is a partition 
$\pi=\{C_1, \dots, C_k\}$ of $N$; 
we write $\pi_i$ to denote the coalition in $\pi$ that contains agent $i$.
An agent $i$ has an {\em NS-deviation} from an outcome $\pi$ if there exists a coalition
$C\in \pi\cup\{\varnothing\}$ such that $C\cup\{i\}\succ_i\pi_i$;
$i$ has an {\em IS-deviation} from an outcome $\pi$ if there exists a coalition
$C\in \pi\cup\{\varnothing\}$ such that $C\cup\{i\}\succ_i\pi_i$ and, 
additionally, $C\cup\{i\}\succeq_j C$ for each $j\in C$. An outcome $\pi$
is {\em Nash stable (NS)} (respectively, {\em individually stable (IS)}) if no agent has
an NS-deviation (respectively, an IS-deviation). An outcome $\pi'$
{\em Pareto dominates} an outcome $\pi$
if $\pi'_i\succeq\pi_i$ for each $i\in N$ and $\pi'_j\succ_j \pi_j$
for some $j\in N$; an outcome is {\em Pareto optimal} if it is not Pareto-dominated
by another outcome. 

Consider a hedonic game $G=(N, (\succeq_i)_{i\in N})$. 
We say that $G$ is {\em dichotomous}
if for every agent $i\in N$  
there exist disjoint sets $\calN^+(i)$ and $\calN^-(i)$
such that $\calN(i)=\calN^+(i)\cup\calN^-(i)$, and
for every pair of coalitions $C, D\in\calN(i)$,
we have $C\sim_i D$ if $C, D\in\calN^+(i)$
or $C, D\in\calN^-(i)$ and $C\succ_i D$
if $C\in\calN^+(i)$, $D\in \calN^-(i)$;
we say that $i$ {\em approves} coalitions in $\calN^+(i)$
and {\em disapproves} coalitions in $\calN^+(i)$.
We say that $G$
is {\em anonymous}
if 
for every agent $i\in N$ 
there exists a weak order $\succeq^*_i$ on $[|N|]$ such that 
for every pair of coalitions $C, D$ we have 
$C\succeq_i D$ if and only if $|C|\succeq^*_i |D|$.

In a {\em hedonic diversity game (HDG)}, 
the set of agents $N$ is partitioned as $N=R\cup B$, 
and each agent $i$ is indifferent between any two coalitions $C, D\in\calN(i)$
that have the same fraction of agents in $R$; we refer to agents in $R$ and $B$
as {\em red} and {\em blue} agents, respectively. 
In such games, the preferences
of every agent $i$ can be described by a weak order $\succeq'_i$ over the set 
$\Theta = \{\frac{j}{k}\mid 0\le j\le |R|, j\le k\le |N|\}$: 
for all $C, D\in \calN(i)$ we have $C\succeq_i D$ if and only if 
$\frac{|C\cap R|}{|C|}\succeq'_i \frac{|D\cap R|}{|D|}$.
We say that a coalition $C$ is {\em homogeneous} if $C\subseteq R$ or $C\subseteq B$.

An anonymous game $(N, (\succeq_i)_{i\in N})$ is said to be {\em single-peaked}
if for every agent $i\in N$ there exists a preferred size $s_i\in [n]$
such that for every pair of coalitions $C, D\in\calN(i)$ with
$|C|<|D|\le s_i$ or $s_i\le |D|<|C|$ it holds that $D\succeq_i C$.
Similarly, a hedonic diversity game $(R\cup B, (\succeq'_i)_{i\in R\cup B})$
is said to be {\em single-peaked} if for every agent $i\in R\cup B$
there exists a preferred value $\rho_i\in\Theta$ such that for every $\rho, \rho'\in\Theta$
such that $\rho_i\le\rho<\rho'$ or $\rho'<\rho\le \rho_i$ it holds that $\rho\succeq'_i \rho'$.

In what follows, we use the fact that the problem {\scshape Exact Cover By 3-Sets} 
({\sc X3C}) is NP-complete \citep{gj}. An instance of this problem is given 
by a set $X=\{1,\dots,m\}$ and a collection $C=\{A_{1}, \dots, A_{k}\}$ 
of 3-element subsets of $X$. It is a yes-instance if there exists a subset 
$C'\subseteq C$ such that $C'$ is a partition of $X$; otherwise it is a no-instance. 
For each $x\in X$, let $J^{x}=\{j^{x}_{1},\dots,j^{x}_{m_{x}}\}$ 
be the set of all indices of sets in $C$ to which $x$ belongs, 
i.e., $ j\in J^{x} $ if and only if $x\in A_{j}$. 

\section{Nash Stability}\label{sec:ns}
\citet{BEI19} show that a hedonic diversity game may fail to have a Nash stable
outcome, even if agents' preferences are single-peaked: indeed, it is easy to see
that an HDG with $|R|=|B|=1$ where the red agent prefers to be alone
and the blue agent prefers to be with the red agent has no NS outcome.
However, Bredereck et al.~do not consider the complexity of checking whether a given 
HDG admits a Nash stable outcome. We will now establish that this problem is NP-complete.
We note that deciding whether an anonymous game has a Nash stable outcome is 
known to be NP-complete as well \citep{Ballester2004}; \cite{peters2016dichotomous} 
shows that the hardness result holds even for anonymous games with dichotomous preferences.

\begin{theorem}\label{thm:ns}
Given an HDG $G = (R\cup B, (\succeq'_i)_{i\in R\cup B})$, it is NP-complete to decide
whether $G$ has a Nash stable outcome. The hardness result holds even if $G$
is dichotomous or if each relation $\succeq'_i$ is a strict order over $\Theta$.
\end{theorem} 
\begin{proof}
To see that this problem is in NP, note that it is possible to check in polynomial time 
whether an outcome $\pi$ is Nash stable by iterating over all agents 
and all coalitions in $\pi$ and 
checking whether the agent prefers joining this coalition to her current coalition.

We start by proving hardness for the case where the agents have 
unrestricted preferences. To this end, we construct a reduction from {\sc X3C}. To begin 
with, we define $f(j)=2j-1-\lfloor\frac{j}{2}\rfloor$, which is a bijection from the 
natural numbers to all numbers not divisible by three. Due to reasons that will become apparent 
in the proof of correctness, it will be convenient to use $f(j)$ instead of $j$ in the construction.

In the following, given $X$ and $C$, we construct the respective hedonic diversity game. 
Intuitively, for each element $x\in X$ we introduce a blue set 
agent $b_{x}$ as well as some red agents so as to map each set $A_{j}$ 
to a unique coalition with ratio
$\frac{2f(j)+3}{2f(j)+6}$ consisting of the three relevant set agents and $2f(j)+3$ 
designated red agents. The game is constructed so that in every Nash stable outcome each 
set agent $b_{x}$ needs to be in a coalition corresponding to a set $A_{j}$ with 
$x\in A_{j}$.

\smallskip

\textbf{Construction:} 
For each element $x\in X$, we introduce one blue set agent $b_{x}$ with the following 
preference relation:
$$ 
b_{x}: \frac{2f(j_{1}^{x})+3}{2f(j_{1}^{x})+6}\sim'_{b_{x}}\dots\sim'_{b_{x}}  
\frac{2f(j_{m_{x}}^{x})+3}{2f(j_{m_{x}}^{x})+6}\succ'_{b_{x}} 0.
$$ 
Moreover, for each $j\in [k],$ 
let us introduce $2f(j)+3$ red filling agents with the following preference relation: 
$$
r_{j}^{p}: \frac{2f(j)+3}{2f(j)+6}\succ'_{r_{j}^{p}} 1, \text{ for all }p\in [2f(j)+3].
$$   
In addition, for each $j\in [m]$, we insert a red stalking agent $z_{j}$ with the following 
preference relation: 
$$
z_{j}: \frac{1}{j+1}\succ'_{z_{j}} 1.
$$

\smallskip

\textbf{Correctness:} 
$(\Rightarrow)$ Assume that there exists a partition $\pi \subseteq C$ of $X$. We
transform $\pi$ into a Nash stable outcome $\pi'$ of the respective hedonic diversity 
game. For each $A_{j}\in \pi,$ the outcome $\pi'$ contains 
a coalition $P_{j}=\{b_{x}\mid x\in A_{j}\}\cup 
\{r_{j}^{p}\mid p\in [2f(j)+3]\}$. All remaining agents are put into singleton 
coalitions and are added to $\pi'$.

We claim that $\pi'$ is Nash stable: By construction, no blue agent has an NS-deviation in 
$\pi'$, as all blue agents are in one of their most preferred coalitions. Moreover, no 
filling agent $r_{j}$, $j\in [k]$, has an NS-deviation, as $r_j$ is either in his 
most preferred coalition or in a singleton coalition. In the latter case, $r_{j}$ can only 
deviate to red singleton coalitions or to coalitions with ratio 
$\frac{2f(i)+3}{2f(i)+6}$ for some $i\in [k]$. However, deviating to a red singleton 
coalition does not make a difference for her and deviating to a coalition with ratio 
$\frac{2f(i)+3}{2f(i)+6}$ for some $i\in [k]$ is never individually rational, as for
every $i\in[k]$ we have $\frac{2f(j)+3}{2f(j)+6}\neq \frac{2f(i)+4}{2f(i)+7}$.
Furthermore, applying the same reasoning, no stalking agent $z_{j}$ has an NS-deviation, as 
for every $i\in[k]$ we have $\frac{2f(i)+4}{2f(i)+7}\neq \frac{1}{j+1}$.

\noindent$(\Leftarrow)$ Let us assume that there exists a Nash stable outcome $\pi$ of the 
hedonic diversity game. Then, no coalition of fraction $\frac{1}{j+1}$ for some $j\in [m]$ 
can be part of $\pi$, as such a coalition is not individually rational for any blue agent. 
Consequently, all stalking agents need to be in homogeneous coalitions in $\pi$. Thereby, $\pi$ 
contains no homogeneous blue coalitions of size $j\in [m]$, as $z_{j}$ would have 
an NS-deviation to it. Therefore, all blue agents need to be in one of their most preferred 
coalitions in $\pi$.

Let $\{P_{1},\dots,P_{t}\}\subseteq \pi$ be the set of all coalitions containing at least one 
set agent. Then, the reasoning above shows that for all $\ell\in [t]$ we have 
$\theta(P_\ell)=\frac{2f(j)+3}{2f(j)+6}$ for some $j\in [k]$. As $\gcd(2f(j)+3,2f(j)+6)=1$ 
for all $j\in [k]$, at least three blue and $2f(j)+3$ red agents are part of coalition 
$P_\ell$. Since there only exist three blue agents for which this ratio is individually 
rational, it follows that all three blue set agents who belong to $A_{j}$ (and no other blue 
agent) are part of $P_\ell$.  Consequently, by removing all red agents from the coalitions in 
$\{P_{1},...,P_{t}\},$ we obtain a cover of $X$ by sets in~$C$.

Note that the reduction still works if no indifferences in the preferences are allowed. In 
this case, it is possible to replace the indifferences in the set agents' preference 
relations by strict preferences in an arbitrary way without affecting the correctness of 
the proof: the first direction still holds, as no blue agent can deviate to a different 
coalition such that the resulting coalition is individually rational for her, since 
$\frac{2f(j)+3}{2f(j)+6}\neq \frac{2f(i)+3}{2f(i)+7},$ for all $i,j\in [k]$. The second 
direction remains unaffected by this change.

Furthermore, note that it is also possible to adapt the reduction so that the resulting 
game is dichotomous and every agent approves at most four ratios: each agent approves 
the individually rational part of her preference relation, while she disapproves all other 
fractions. As {\sc X3C} remains NP-complete if we require that every element appears 
in at most three sets \citep{gj}, our problem remains hard if
each agent approves at most four fractions. The proof of 
correctness still goes through in this case, as its first part remains unaffected, while in the 
second part one only has to add the observation that all blue agents need to be in one of 
their approved coalitions in every Nash stable outcome, as $0$ is among their approved 
fractions.
\end{proof}

To complement the hardness result of Theorem~\ref{thm:ns}, we show
that deciding the existence of a Nash stable outcome is in P if 
$\min\{|R|, |B|\}$ is bounded by a constant.

\begin{theorem}\label{thm:ns-xp}
Given an HDG $G = (R\cup B, (\succeq'_i)_{i\in R\cup B})$ with $|R\cup B|=n$,
$\min\{|R|, |B|\}=p$, we can decide
if $G$ has a Nash stable outcome in time $(np)^{p}\cdot\mathrm{poly}(n)$.
\end{theorem}
\begin{proof}
Assume without loss of generality that $|B|\le |R|$, so $p = |B|$.
Our algorithm proceeds in three stages. 

In the first stage, we guess a partition $\{C_1, \dots, C_k\}$ of $B$, 
together with $k$ positive integers $n_1, \dots, n_k$ such that $n_i\ge |C_i|$
and $n_1+\dots+n_k\le n$; let $C_0=\varnothing$ and $n_0=n-(n_1+\dots+n_k)$. 
We will try to construct
an NS partition $\{A_1, \dots, A_k, D_1, \dots, D_\ell\}$ for some $\ell\ge 0$
so that $C_i\subseteq A_i$, $|A_i|=n_i$ for each $i\in [k]$,
and $D_i\subseteq R$ for $i\in [\ell]$. Note that we can enumerate all possible
guesses in time $O(p^p\cdot n^p)$. 

In the second stage, given a guess $\{C_1, \dots, C_k\}$, $(n_1, \dots, n_k)$, we construct
an instance of the network flow problem as follows. We create a source, 
a sink, a node $x_i$ for each $i\in R$ and a node $y_j$ for each $j=0, \dots, k$.
The source is connected to all nodes $x_i$, $i\in R$, by an edge of capacity $1$, and
each node $y_j$, $j=0, \dots, k$, is connected to the sink by an edge of capacity 
$n_j-|C_j|$. 
Further, there is an edge of capacity $1$ from $x_i$ to $y_j$, $j\in [k]$, 
if and only if $\frac{n_j-|C_j|}{n_j} \succeq'_i \frac{n_s-|C_s|+1}{n_s+1}$ 
for each $s\in [k]\setminus\{j\}$ and, moreover, $\frac{n_j-|C_j|}{n_j}\succeq'_i 1$.
Finally, there is an edge of capacity $1$ from $x_i$ to $y_0$ if 
$1\succeq'_i \frac{n_s-|C_s|+1}{n_s+1}$ for all $s\in [k]$.

Intuitively, an edge from $x_i$ to $y_j$, $j\in [k]$, indicates that $i$ weakly prefers
being in a coalition of size $n_j$ with $|C_j|$ blue agents to deviating to 
a coalition of size $n_s$ with $|C_s|$ blue agents, for $s\neq j$, 
or to a homogeneous red coalition; an edge from $x_i$ to $y_0$ indicates that $i$
weakly prefers a homogeneous red coalition to other available options.
By construction, there is a flow of size $|R|$
in this network if and only if there is a partition 
$\{A_1, \dots, A_k, D_0\}$ of $R\cup B$ such that
$C_i\subseteq A_i$, $|A_i|=n_i$ for $i\in [k]$ and no red
agent has an NS-deviation from this partition. 
Thus, if this instance of network flow does not admit a flow of size $|R|$, 
we reject the current guess, and otherwise we proceed to the next stage.

In the third stage, we split $D_0$ into $D_1, \dots, D_\ell$ for some $\ell\ge 0$. 
Note that, no matter how we do this, if no red agent had an NS-deviation in
$\{A_1, \dots, A_k, D_0\}$, this will also be the case for the new partition. Thus, 
we can focus on the blue agents. 
Given a $t\in [n_0]$, we say that $t$ is {\em safe
for an agent $j\in B$} if $j$ weakly prefers her current coalition in $\{A_1, \dots, A_k, D_0\}$
to a coalition consisting of herself and $t$ red agents;
we say that $t$ is {\em safe} if it is safe for each $j\in B$. Let $T\subseteq [n_0]$
be a collection of all safe integers. Then, we can subdivide $D_0$ so that in the resulting
partition no blue agent has an NS-deviation if and only if $n_0$ can be represented 
as a sum of integers from $T$; the latter problem is a variant
of {\sc Knapsack}, and can be solved by dynamic programming in time $O(n^2)$.  
\end{proof}
Theorem~\ref{thm:ns-xp} puts the problem of finding a Nash stable outcome
in the complexity class XP with respect to the parameter $p$; however, we
do not know if this problem is fixed-parameter tractable (FPT) with respect to this parameter.

For HDGs with dichotomous preferences,
another natural parameter is the number of ratios approved
by each agent. However, our problem turns out to be 
para-NP-hard with respect to this parameter:
the hardness proof in Theorem~\ref{thm:ns}
goes through even if each agent only approves at most four ratios in $\Theta$. 
Similarly, \citet{peters2016dichotomous} shows that finding a Nash stable outcome 
in dichotomous anonymous games remains NP-hard if each agent 
approves at most four coalition sizes. Interestingly, we can prove that
the latter problem becomes polynomial-time solvable if each agent
approves at most one coalition size, but it is not clear how to extend
this proof to dichotomous HDGs.

\begin{theorem}\label{thm:anon}
Given a dichotomous anonymous game $G = (N, (\succeq_i)_{i\in N})$ 
where for each agent $i\in N$ the set $\calN^+(i)$ is of the form
$\{C\in\calN(i): |C|=s_i\}$ for some $s_i\in\mathbb{N}$,
we can decide in polynomial time whether $G$ admits a Nash stable outcome.
\end{theorem}
\begin{proof}
If there does not exist an agent $i\in N$ with $s_i=1$, 
it follows that the grand coalition is Nash stable. Thus, in the following, 
we assume that there exists at least one agent approving a singleton coalition. 
For each $j=0, 1, \dots, n$, let $N_j$ be the set of all agents $i \in N$ with $s_i=j$.
Moreover, let $\ell=\max\{i \mid N_j\neq\varnothing\text{ for all }j\in [i]\}$. 

We claim that for each $j\in [\ell]$ all agents in $N_j$ need to be in coalitions of size 
$j$ in every Nash stable outcome. This follows easily by induction on $j$. Indeed, if an 
agent in $N_1$ is not in a singleton coalition, she can always deviate to being in a singleton 
coalition. Further, if $1<j\le \ell$ and for each $k\in [j-1]$ all agents in $N_k$ are in 
coalitions of size $k$, it follows by definition of $\ell$ that there needs to exist a 
coalition of size $k$ for all $k\in [j-1]$. Now, assuming that an agent in $N_j$ is not in a 
coalition of size $j$, she can always deviate to a coalition of size $j-1$.

For each $j\in [\ell]$, let
$d_j = j\times \left\lceil\frac{|N_j|}{j}\right\rceil - |N_j|$. 
Intuitively, agents in $N_j$ need $d_j$ additional agents to split into coalitions of size $j$.
Let $d=\sum_{j=1}^\ell d_j$ and set $N' = N_0\cup \bigcup_{j>\ell}N_j$.

If $d\le |N'|$, we can construct a Nash stable outcome as follows: we put the agents from 
$N'$ into $\ell$ pairwise disjoint sets $D_1, \dots, D_\ell \subseteq N'$ so that $|D_j|=d_j$ 
for each $j\in [\ell]$. Then, by construction, $|N_j\cup D_j|$ is divisible by $j$ for each 
$j\in [\ell]$. Hence, we can split agents in $N_j\cup D_j$ into coalitions of size $j$. The 
remaining agents in $N'$ are placed in singleton coalitions. This partition is Nash stable, 
since all agents in $N\setminus N'$ approve their coalition sizes and none of the agents in 
$N'$ approves coalition sizes $1, \dots, \ell,$ or $\ell+1$.

On the other hand, if $d>|N'|$, it follows that there cannot exist a Nash stable outcome: indeed, in this case, there are not enough agents in $N'$ to construct an outcome in which each agent in $N\setminus N'$ is in a coalition of her approved size, and we have argued that this is necessary for stability.
\end{proof}

\section{Individual Stability}\label{sec:is}
\citet{BEI19} describe an algorithm that, given an HDG with single-peaked preferences,
outputs an individually stable outcome in polynomial time. This algorithm is fairly complex:
its description and analysis take up almost four pages of the AAMAS'19 paper. 
It is similar in spirit to the algorithm of \citet{Bogomolnaia2002} that finds
an IS outcome of an anonymous game with single-peaked preferences in polynomial time. The main
contribution of this section is a much simpler polynomial-time algorithm that can find
an individually stable outcome of {\em any} HDG; this result is particularly surprising, 
because it is known that not every anonymous game admits 
an IS outcome \citep{Bogomolnaia2002}.

\begin{algorithm}[t] 
	\caption{Computing an individually stable outcome}
	\label{alg:IS}
	\begin{algorithmic}[1]
		\Require{HDG $G=(R\cup B,(\succsim'_{i})_{i\in R\cup B})$} 
		\Ensure{Individually stable outcome $\pi$ of game $G$}
		\State {Let $B^{*}=\{b^*_1, \dots, b^*_k\} = \{b\in B: \frac{1}{2}\succsim'_b 0\}$;}
		\State {Let $R^{*}=\{r^*_1, \dots, r^*_\ell\} = \{r\in R: \frac{1}{2}\succsim'_r 1\}$;}
		\State{Let $C_0=\{r^{*}_{1},\dots,r^{*}_{\min \{k,\ell\}}\}\cup 
                              \{b^{*}_{1},\dots,b^{*}_{\min \{k,\ell\}}\}$;}
		\State{Let $C=C_0$;}
		\Repeat
		\For{$i\in N\setminus C$ }                    
		\If{$i$ has an IS-deviation from $\{i\}$ to $C$}
		\State $C=C\cup\{i\}$;
		\EndIf
		\EndFor
		\Until{$C$ has not changed in the previous iteration}
		\State{Let $C_1=C$;}
		\State \Return {$\pi=\{\{i\}\mid i\in N\setminus C_1 \}\cup \{C_1\}$;}
	\end{algorithmic}
\end{algorithm}

\begin{theorem}\label{thm:is-easy}
Given an HDG $G = (R\cup B, (\succeq'_i)_{i\in R\cup B})$, we can compute an individually 
stable outcome of $G$ in polynomial time.
\end{theorem}
\begin{proof}
In what follows, we say that a coalition $C\subseteq R\cup B$
is {\em balanced} if $|C\cap R|=|C\cap B|$.

We claim that Algorithm~\ref{alg:IS} outputs an IS outcome of $G$.
The first phase of this algorithm creates a maximum-size balanced coalition 
$C$ such that all agents
in $C$ prefer $C$ to being alone; all other agents
are placed in singleton coalitions.
In the second phase, the algorithm checks if any of the remaining agents 
has an IS-deviation to $C$; if yes, some such agent is invited to join $C$.
This step is repeated until none of the remaining agents has an IS-deviation to $C$.
To avoid ambiguity, we use $C_0$ and $C_1$ to denote the `large' coalition obtained
at the end of the first phase and at the end of the second phase, respectively.

Consider the sets $B^*$ and $R^*$ defined by our algorithm, and assume without
loss of generality that $|B^*|\ge |R^*|$. 
By construction, $C_0$ contains all red agents who weakly prefer being 
in a balanced coalition to being in a homogeneous coalition. 
In particular, this means that no red agent in $N\setminus C_1$ 
would allow a blue agent to join her singleton coalition.

To prove that the partition computed by our algorithm 
is individually stable, we consider three classes of agents:

\begin{description}
\item[Agents in $\boldsymbol{C_0}$:] 
By deviating, these agents can form a homogeneous coalition or a balanced coalition. 
They weakly prefer $C_0$ to a homogeneous coalition, 
and, since they approve all subsequent changes to $C$, they weakly prefer $C_1$ to $C_0$
(which is balanced). Thus, they have no IS-deviation.
\item[Agents in $\boldsymbol{N\setminus C_1}$:]
By construction, these agents do not have an IS-deviation to $C_1$, and they are indifferent
between being alone and joining another agent of the same color. 
Further, a deviation that results in a two-agent balanced coalition is not an 
IS-deviation: all red agents in $N\setminus C_1$ strictly prefer being alone 
to being in a balanced coalition. Thus, agents in $N\setminus C_1$ have no IS-deviation.
\item[Agents in $\boldsymbol{C_1\setminus C_0}$:]
Joining $C$, these agents strictly preferred $C$ to being in a homogeneous coalition, 
and they have approved all changes to $C$ since then. 
Thus, they strictly prefer being in $C_1$ to being in a homogeneous
coalition. Further, a blue agent in $C_1\setminus C_0$ cannot join
a singleton red coalition, since the red agent strictly prefers to be left alone. 
On the other hand, every red agent in $C_1\setminus C_0$ strictly prefers
being alone to being in a balanced coalition, so by transitivity 
she strictly prefers staying in $C_1$ to joining a singleton blue coalition. 
\end{description}
\vspace*{-0.8cm}
\end{proof} 

The following example illustrates the execution of our algorithm.

\begin{example}\label{ex:is-good}
Consider an HDG with $B=\{1, 2, 3\}$, $R=\{4, 5\}$.
The agents' preferences over $\Theta$ are given by 
\begin{align*}
\frac{2}{3}\sim'_1\frac{1}{2}\succ'_1 0, \qquad 
\frac{2}{3}\succ'_2\frac{1}{2}\succ'_2 0,\\
\frac{1}{4}\succ'_3 0, \qquad
\frac{2}{3}\succ'_4\frac{1}{2}\succ'_4 1, \qquad
\frac{2}{3}\succ'_5 1.
\end{align*}
The algorithm sets $B^*=\{1, 2\}$, $R^*=\{4\}$, and $C_0=\{1, 4\}$.
Then, in the second phase agent $5$ has an IS-deviation to $C_0$:
we have $\{1, 4, 5\}\succ_5 \{5\}$, 
$\{1, 4, 5\}\succ_4 \{1, 4\}$
and $\{1, 4, 5\}\sim_1 \{1, 4\}$.
Since agents $2$ and $3$ do not have an IS-deviation to $\{1, 4, 5\}$, 
the algorithm stops and outputs $(\{1, 4, 5\}, \{2\}, \{3\})$.
\end{example}
In Example~\ref{ex:is-good}, our algorithm outputs a Pareto optimal outcome;
however, our next example shows that this is not always the case.

\begin{example}\label{ex:is-bad}
Consider an HDG with $B=\{1, 2, 3, 4\}$, $R=\{5\}$, where
$\frac{1}{5}\succ'_i 0$ for each $i\in B$, and 
the red agent's most preferred ratio is $\frac{1}{5}$.
Then, for every agent the formation of the grand coalition is the most 
preferred outcome. However, as $B^*=\varnothing$, the algorithm sets
$C_0=\varnothing$ and hence outputs the partition consisting of five singletons.
\end{example}

Interestingly, the algorithm of Bredereck et al.~would output the grand coalition 
on the instance from Example~\ref{ex:is-bad}.
However, it is not the case that on any single-peaked instance the output 
of Bredereck et al.'s algorithm Pareto-dominates the output of Algorithm~\ref{alg:IS}; 
there exists an example where the converse is true. This means, in particular, that neither
of these algorithms is guaranteed to output a Pareto-optimal outcome on single-peaked
instances; this is in contrast with the algorithm of \citet{Bogomolnaia2002}, 
which, given a single-peaked anonymous game, always outputs an IS outcome that
is also Pareto optimal.

\paragraph{Better-Response Dynamics}
We note that one can also attempt to reach an IS outcome by a sequence of IS deviations:
starting from an arbitrary partition, we check if the current partition
is individually stable, and if not, we pick an agent who has an IS-deviation 
and allow her to perform it. 
We will refer to this procedure as the {\em better-response dynamics (IS-BRD)}.
While IS-BRD is a very general algorithm that can be used for arbitrary hedonic games, 
it may fail to converge even if an IS outcome exists, and it may need a 
superpolynomial number of iterations to converge \citep{Gairing2010}; 
however, no results concerning its convergence are known
for hedonic diversity games or for anonymous games. 


The following proposition, which applies to arbitrary dichotomous hedonic games, 
makes partial progress towards the understanding of the performance of IS-BRD: 
it shows that for such games IS-BRD always converges as long as 
in the initial partition all agents belong to the grand coalition.
We note that existence of IS outcomes in dichotomous games has been established by 
\citet{peters2016dichotomous}; his proof provides a polynomial-time algorithm
for finding an IS outcome, but this algorithm differs from IS-BRD.

\begin{proposition}
For every dichotomous hedonic game with a set of agents $N$,
any sequence of IS deviations starting
from the grand coalition converges to an IS outcome after at most $|N|$ iterations.  
\end{proposition}
\begin{proof}
If the grand coalition is individually stable, we are done. Otherwise, 
let $\pi^{(k)}$ be the partition that forms after a sequence of $k$ deviations, $k\ge 1$, 
and let $N^{(k)}$ be the set of all agents that have not deviated during this process.

We claim that the agents in $N^{(k)}$ form a coalition in $\pi^{(k)}$, and
all agents in $N\setminus N^{(k)}$ approve their coalitions in $\pi^{(k)}$ 
(and hence only agents in $N^{(k)}$ can perform IS deviations).

This follows by induction on $k$. Indeed, if $k=1$, 
then the first agent who deviates forms a singleton coalition she approves,
and all remaining agents stay together in one coalition. Now, suppose our statement is true for 
$k$; we will argue that it is true for $k+1$. Consider the agent $j$ 
who performs the $(k+1)$-st deviation. By the inductive hypothesis, we have 
$j\in N^{(k)}$, so $N^{(k+1)}=N^{(k)}\setminus\{j\}$, i.e., agents in $N^{(k+1)}$
form a coalition in $\pi^{(k+1)}$. Now, consider a coalition $C$ in $\pi^{(k)}$ formed 
by agents in $N\setminus N^{(k)}$. By the inductive hypothesis, all agents in $C$
approve $C$. If $C$ is unaffected by the deviation, this remains to be the case.
On the other hand, if $j$ joins $C$, $j$ approves $C\cup\{j\}$; since all agents in $C$ do not object to $j$
joining them, it follows that they, too, approve $C\cup\{j\}$. This establishes
our claim.

It follows that every agent can deviate at most once: after her first deviation,
she approves the coalition she is in and does not want to move again. Hence, 
there can be at most $|N|$ deviations.
\end{proof}

We will revisit IS-BRD in Section~\ref{sec:emp}, where we compare
different approaches to finding IS outcomes in HDGs.


\section{Diversity Games With $\boldsymbol{k}$ Classes}\label{sec:ktypes}
\citet{BEI19} define hedonic diversity games for two agent classes. However, 
diversity-related considerations remain relevant in the presence
of three or more classes: for instance, in the example discussed 
in the beginning of the paper, the visiting students may come from several 
different countries. To capture such settings, we need to reason about games 
with $k$ agent classes for $k>2$:
e.g., for $k=3$ we may have red, blue, and green agents. 
A direct generalization of the model of Bredereck et al.~is
to allow a red agent to base her preferences on both the ratio of red and blue agents 
and the ratio of red and green agents; a more restrictive approach 
is to assume that a red agent only cares about the fraction of red agents 
in her coalition. We will now explore both of these approaches; as we feel that the latter
approach is closer in spirit to the original HDG model, we reserve the term {\em $k$-HDG}
to refer to the more restrictive model and refer to games where agents can have
arbitrary preferences over ratios as {\em $k$-tuple HDGs}. 

\begin{definition}\label{def:k-first}
A {\em $k$-tuple hedonic diversity game} is a hedonic game 
$(N, (\succeq_i)_{i\in N})$ where $N$ can be partitioned into $k$
pairwise disjoint sets $R_1, \dots, R_k$ so that for each $j\in [k]$, 
each agent $i\in R_j$,
and every pair of coalitions $C, D\in\calN(i)$ 
with $\frac{|C\cap R_s|}{|C\cap R_j|}=\frac{|D\cap R_s|}{|D\cap R_j|}$
for each $s\in [k]$ we have $C\sim_i D$.
\end{definition}
Given an $n$-agent $k$-tuple HDG and an agent $i\in R_j$, 
we can map a coalition $C\in \calN(i)$
to a $k$-tuple of fractions $\left(\frac{|C\cap R_\ell|}{|C\cap R_j|}\right)_{\ell\in [k]}$:
$i$ is indifferent between two coalitions that map to the same tuple. 
Hence, $i$'s preferences can be described by a partial order over such tuples.
As the number of tuples of this form is bounded by $n^{2k}$,  
if $k$ is bounded by a constant, the size of this representation
is polynomial in $n$. On the other hand, 
every hedonic game with $n$ agents is an $n$-tuple HDG,
as we can simply place each agent in a separate class. 

Our XP result for Nash stability extends to this more general setting:
if the total size of the smallest $k-1$ classes can be bounded by a constant, 
a Nash stable outcome can be found in time polynomial in the input representation size,
which, in turn, can be bounded as $\mathcal{O}(n^3)$ in this case.
The proof (omitted) is
a 
simple generalization of 
the proof of Theorem~\ref{thm:ns}. 

\begin{theorem}\label{thm:ns-k}
Given a $k$-tuple HDG $G = (N, (\succeq_i)_{i\in N})$
with classes $R_1, \dots, R_k$ where $|N|=n$ and $\min_{j\in [k]}|N\setminus R_j| = p$,
we can decide
whether $G$ has a Nash stable outcome in time $(np)^{p}\cdot\mathrm{poly}(n)$.
\end{theorem}

In contrast, the following example shows that Theorem~\ref{thm:is-easy} 
does not extend to $k>2$, i.e., $k$-tuple HDGs may fail to have an individually stable outcome 
if $k>2$. 

\begin{example}\label{ex:kHDG}
Consider a $3$-tuple HDG $G$ with $N=R_1\cup R_2\cup R_3$, $R_1=\{r\}$, $R_2=\{b\}$, 
$R_3=\{g\}$, where the agents have the following preferences:
\begin{align*}
\{r, b\} &\succ_r \{r, g\}\succ_r \{r\} \succ_r \{r, b, g\};\\
\{b, g\} &\succ_b \{b, r\}\succ_b \{b\} \succ_b \{r, b, g\};\\
\{g, r\} &\succ_g \{g, b\}\succ_g \{g\} \succ_g \{r, b, g\};
\end{align*}
It is immediate that this game has no IS outcome. Indeed, every agent
has an IS-deviation from the grand coalition, and if each agent is in a singleton 
coalition, $r$ has an IS-deviation to $\{b\}$. Further, if an outcome consists of a coalition of size
$2$ and a singleton coalition, for one of the agents in the coalition of size $2$
joining the singleton agent is an IS-deviation.
To extend this example to $k>3$ classes, we can simply add agents
who prefer to be in homogeneous coalitions.
\end{example}
Moreover, for every fixed $k\ge 5$, deciding the existence of an IS outcome 
in a $k$-tuple HDG is NP-complete.

\begin{theorem}\label{thm:is-k}
Given a $k$-tuple HDG $G = (N, (\succeq_i)_{i\in N})$
with $k\ge 5$, it is NP-complete to decide
whether $G$ has an individually stable outcome. 
\end{theorem}
\begin{proof}
To prove membership, note that it is possible to check in polynomial time whether an 
outcome is individually stable by iterating over all agent-coalition pairs and checking 
whether the agent has an IS-deviation to the coalition.

To prove hardness, we construct a reduction from {\sc X3C}. We map a pair 
$(X, C)$ to a diversity game with 5 classes of agents 
(\textbf{r}ed, \textbf{b}lue, \textbf{g}reen, \textbf{y}ellow, \textbf{w}hite).

The general idea of the construction is to map 
each element $x\in X$ to a blue set agent $b_{x}$ such that in 
an individually stable outcome $b_{x}$ needs to be 
in a coalition with ratio $(\frac{1}{j+3}, \frac{j+2}{j+3},0,0,0)$ 
for some $j\in [k]$ such that $x\in A_{j}$. 
This is ensured by an IS-penalizing component `destabilizing' every outcome for 
which this is not the case.
This enables us to map each set $A_{j}\in C$ to a unique coalition 
with ratio $(\frac{1}{j+3}, \frac{j+2}{j+3},0,0,0)$. 
\smallskip

\textbf{Construction: }
For each element $x\in X$, we introduce one blue set agent $b_{x}$ 
with the following preference relation $\succeq'_{b_x}$ over tuples of fractions:
\begin{align*}
b_{x} : &\left(\frac{1}{j_{1}^{x}+3}, \frac{j_{1}^{x}+2}{j_{1}^{x}+3},0,0,0\right)\sim'_{b_{x}} 
\dots \\
&\sim'_{b_{x}} 
\left(\frac{1}{j_{m_{x}}^{x}+3},\frac{j_{m_{x}}^{x}+2}{j_{m_{x}}^{x}+3},0,0,0\right)\succ'_{b_{x}}
\left(0,\frac{1}{3},\frac{1}{3},\frac{1}{3},0\right)\\
&\succ'_{b_{x}} \left(0,\frac{1}{2},\frac{1}{2},0,0\right) 
\succ'_{b_{x}}  \left(0,1,0,0,0\right).
\end{align*}  
Moreover, for each $j\in [k],$ we introduce a red fraction agent $r_{j}$ 
with the following preference relation $\succeq'_{r_j}$: 
$$
r_{j}: \left(\frac{1}{j+3},\frac{j+2}{j+3},0,0,0\right)\succ'_{r_{j}} \left(1,0,0,0,0\right).
$$  
Additionally, for each $j\in [k]$, we introduce $j-1$ redundant 
blue agents with the following preference relations: 
$$
b_{j}^{p}: (\frac{1}{j+3},\frac{j+2}{j+3},0,0,0)\succ'_{b_{j}^{p}} (0,1,0,0,0), 
\text{ } \forall p\in [j-1].
$$  
Finally, we insert an IS-penalizing component consisting of one green, one yellow, 
and one white agent: 
\begin{align*}
g : & \left(0,0,\frac{1}{2},0,\frac{1}{2}\right) \succ'_{g} 
\left(0,\frac{1}{3},\frac{1}{3},\frac{1}{3},0\right) \\ 
&\succ'_{g} \left(0,\frac{1}{2},\frac{1}{2},0,0\right)\succ'_{g}  \left(0,0,1,0,0\right), \\
y : & \left(0,\frac{1}{3},\frac{1}{3},\frac{1}{3},0\right)\succ'_{y} 
\left(0,0,0,\frac{1}{2},\frac{1}{2}\right)\succ'_{y} \left(0,0,0,1,0\right), \\
w : & \left(0,0,0,\frac{1}{2},\frac{1}{2}\right) \succ'_{w} 
\left(0,0,\frac{1}{2},0,\frac{1}{2}\right)\succ'_{w} \left(0,0,0,0,1\right).
\end{align*}

\textbf{Correctness: }$(\Rightarrow)$ 
Assume that there exists a partition $\pi \subseteq C$ of $X$. We will transform 
$\pi$ into an individually stable outcome $\pi'$ of the respective $5$-tuple hedonic 
diversity game. For each $A_{j}\in \pi,$ the partition $\pi'$ contains a coalition 
$P_{j}=\{b_{x}\mid x\in A_{j}\}\cup \{r_{j}\}\cup \{b_{j}^{p}\mid p\in [j-1]\}$.  
We refer to these coalitions as the first part of $\pi'$. 
Further, we put all remaining red and blue agents 
in singleton coalitions, add them to $\pi'$, and refer to them as the second part of 
$\pi'$. Finally, we add the coalitions $\{g\}$ and $\{y,w\}$ to $\pi'$.

In the following, we prove that no agent has an IS-deviation in $\pi'$ by iterating over all 
coalitions and arguing that no agent has an IS-deviation to this coalition. First, every 
coalition in the first part of $\pi'$ is the unique most preferred coalition of its fraction agent. 
Thereby, no fraction agent accepts a deviation to her coalition. Second, no agent has an 
IS-deviation to a coalition from the second part, as for no agent in the second part being in 
a coalition where players of her class account for $\frac{1}{2}$ is individually rational. 
Third, no agent has an IS-deviation to $\{y,w\}$, as this is $w$'s most preferred 
coalition. Fourth, no agent wants to deviate to $\{g\}$, as all set agents and $w$ are in 
one of their most preferred coalitions. Fifth, no agent has an IS-deviation to being in a 
singleton coalition, as $\pi'$ is individually rational. Consequently, $\pi'$ is individually 
stable.

$(\Leftarrow)$ Let us assume that there exists an individually stable outcome $\pi$ of the 
hedonic diversity game. We claim that every blue set agent is in one of her most preferred 
coalitions in $\pi$. For the sake of contradiction, assume that there exists some set 
agent $p$ who is not in one of her most preferred coalitions. 
In the following, we iterate over all possible individually rational coalitions 
that $g$ can be part of, and show that $\pi$ cannot be individually stable if it contains this 
coalition.

If $\{g\}\in \pi$, $p$ needs to be in a homogeneous coalition, 
as this is her only remaining individually rational coalition. 
Thereby, $p$ has an IS-deviation to $\{g\}$. 
If $\{g,p\}\in \pi$, agent $y$ has an IS-deviation to this coalition. 
If $\{g,y,p\} \in \pi$ then $w$ needs to be in a singleton coalition, 
as this is her only remaining individually rational coalition. 
Thereby, $g$ has an IS-deviation to $\{w\}$. 
If $\{g,w\} \in \pi$, $w$ has an IS-deviation to the coalition $\{y\}$, 
as $\{y\}$ is the only remaining individually rational coalition for $y$.

Consequently, there is no individually stable outcome where a set agent is not in one of her most 
preferred coalitions. Let $\{P_{1},...,P_{t}\}\subseteq \pi$ be the set of all 
coalitions containing at least one set agent. Then, for all $\ell\in [t]$ it holds that 
$\theta(P_\ell)=(\frac{1}{j+3}, \frac{j+2}{j+3},0,0,0)$ for some $j\in [k]$. As 
$(\frac{1}{j+3}, \frac{j+2}{j+3},0,0,0)$ is only individually rational for exactly one red 
agent, $P_\ell$ needs to consist of one red and $j+2$ blue agents. Now, recall 
that this tuple is individually rational for exactly $j+2$ blue agents including the 
three corresponding set agents. It follows that all three blue set agents who belong to 
$A_{j}$ and no other set agents are part of $P_\ell$.  Consequently, by removing all non-blue 
agents from the coalitions in $\{P_{1},...,P_{t}\},$ we obtain a cover of $X$ 
by sets in $C$.

Note that the proof also goes through if no indifferences in the agents' preferences are allowed, as 
it is possible to replace the indifferences in the set agents' preferences by strict 
preferences in an arbitrary way without affecting the validity of the proof.
\end{proof}

Now, in $k$-tuple HDGs, an agent may have very complex preferences over ratios
of agents from different classes. 
In our second model, an agent does not distinguish among classes other than her own
and hence only cares about the fraction of members of her class 
in her coalition; we will refer to these games as $k$-HDGs. We note that
a similar approach has been used recently to provide a game-theoretic
model of Schelling segregation with $k>2$ agent classes \citep{jump,echzell}; 
this line of work, 
while similar in spirit to hedonic diversity games, is, however, 
very different from a technical perspective. 

\begin{definition}\label{def:k-second}
A {\em hedonic diversity game with $k$-classes ($k$-HDG)} is a hedonic game 
$(N, (\succeq_i)_{i\in N})$ where $N$ can be partitioned into $k$
pairwise disjoint sets $R_1, \dots, R_k$ so that for each $j\in [k]$, 
each agent $i\in R_j$,
and every pair of coalitions $C, D\in\calN(i)$ we have $C\sim_i D$
as long as $\frac{|C\cap R_j|}{|C|}=\frac{|D\cap R_j|}{|D|}$.
\end{definition}
By definition, every $k$-HDG is a $k$-tuple HDG, but the converse is not true,
e.g., the game in Example~\ref{ex:kHDG} is not a $k$-HDG for any value of $k$.
Further, a $2$-HDG is simply an HDG as defined by Bredereck et al.
Unlike $k$-tuple HDGs, $k$-HDGs admit a succinct representation even if $k$ is not bounded
by a constant: for each agent $i\in R_j$, her preferences over coalitions in $\calN(i)$ 
can be described by her preferences over fractions of the form $\frac{\ell}{r}$, 
where $r\in [n]$, $\ell\in[\min\{|R_j|, r\}]$.

Remarkably, this formalism captures anonymous games.

\begin{proposition}
Every anonymous game can be represented as a $k$-HDG.
\end{proposition}
\begin{proof}
Recall that an anonymous game $G = (N, (\succeq_i)_{i\in N})$
with $|N|=n$ can be equivalently represented by a collection $(\succeq^*_i)_{i\in N}$ 
of weak orders over $[n]$: $C\succeq_i D$ if and only if $|C|\succeq^*_i |D|$.
It follows that $G$ can be viewed as an $n$-HDG with partition
$N=R_1\cup\ldots\cup R_n$, so that $R_i=\{i\}$ for each $i\in N$. Indeed, for each 
$i\in N$ and each coalition $C\in\calN(i)$,
the fraction of agents from class $R_i$ in $C$ is exactly $\frac{1}{|C|}$,
so two coalitions $C, D\in\calN(i)$ have the same fraction of agents from 
$R_i$ if and only if they have the same size.
\end{proof}

It follows that $k$-HDGs inherit negative results for anonymous games, 
such as non-existence of IS outcomes \citep{Bogomolnaia2002} and 
hardness of deciding whether a given game has an IS outcome \citep{Ballester2004}.

\begin{corollary}\label{cor:shdg}
There exists a $k$-HDG that has no individually stable outcome. Moreover, 
deciding if a given $k$-HDG has an individually stable outcome is NP-complete.
\end{corollary}

Now, \citet{Bogomolnaia2002} show that every single-peaked anonymous game 
has an IS outcome. The definition of single-peaked preferences extends 
naturally to $k$-HDGs, e.g., 
we can use essentially the same definition as for HDGs. However, it is not clear
if the algorithm of \citet{Bogomolnaia2002} for finding an IS outcome in single-peaked anonymous 
games can be extended to single-peaked $k$-HDGs; this is an interesting question
for future work. Note also that the hardness result of Corollary~\ref{cor:shdg}
only holds for $k=n$; it is not clear if the problem of finding an IS outcome
in $k$-HDGs remains hard for small values of $k$ (e.g., $k=3$).
Further, the anonymous game with no IS outcomes constructed by 
\citet{Bogomolnaia2002} has 63 agents and therefore
translates into a $63$-HDG; it remains an open problem 
whether $k$-HDGs with $k<63$ are guaranteed to have an IS-outcome.

\section{Empirical Analysis}\label{sec:emp}
The algorithm for computing individually stable outcomes described in 
Section~\ref{sec:is} has two substantial advantages over the algorithm of Bredereck et al.:
first, it works for general preferences, and second, it is much simpler. However, 
as illustrated by Example~\ref{ex:is-bad}, Bredereck et al.'s algorithm may result
in higher agents' satisfaction. The IS-BRD algorithm is even simpler, but 
we do not know if it always converges to an IS outcome. To better understand the
performance of these three algorithms, in this section, we empirically compare them 
with respect to three measures: the average social welfare,
the average coalition size and the average diversity. 

\smallskip

\noindent{\bf Preference Models\ }
As Bredereck et al.'s algorithm is only defined for single-peaked HDGs, 
we only use single-peaked instances in our analysis. We consider three 
intuitively appealing ways
of sampling preferences over ratios that are single-peaked on $\Theta$.
\begin{description}
\item
{\bf Uniform single-peaked preferences (uSP).} For each agent $i$, we sample $\succ'_i$
uniformly at random among all single-peaked strict orders on $\Theta$, using the algorithm of
\citet{Walsh} (see also \citet{LL17}).
\item
{\bf Uniform-peak single-peaked preferences (upSP).} To generate $\succ'_i$, 
we first select $i$'s most preferred ratio 
by choosing a point in $\Theta$ uniformly at random.
We then continue to place elements of $\Theta$ in positions
$2, \dots, |\Theta|$ of $\succ'_i$ one by one; when $k<|\Theta|$ elements have been ranked, 
there are at most two elements of $\Theta$ that can be placed in position $k+1$
so that the resulting ranking is single-peaked on $\Theta$, and we choose between them
with equal probability. This approach to sampling single-peaked preferences
was popularized by \citet{ConitzerSP}.
\item
{\bf Symmetric single-peaked preferences (symSP).} For each agent $i$, we choose her preferred
point $\theta_i$ from the uniform distribution on $[0, 1]$ and define the relation 
$\succeq'_i$ so that $\theta\succeq'_i \theta'$ if and only if 
$|\theta-\theta_i|\le |\theta'-\theta_i|$. While theoretically the resulting
relation may have ties, in our experiments this approach always generated strict orders.
\end{description}
The first two distributions are quite different from each other,
e.g., in the uSP model a ranking where $0$ appears first is exponentially less likely 
than a ranking where $\frac12$ appears first, while in the upSP model
we are equally likely to see $0$ and $\frac12$ ranked first. 
On the other hand, upSP and symSP appear to be fairly similar, 
but our experiments show that our algorithms behave differently on them. 

\smallskip

\noindent{\bf Performance measures\ }
The primary measure we are interested in is the social welfare, i.e., 
the sum of agents' utilities. However, in general, 
this measure is difficult to define, since in HDGs
agents' preferences over coalitions are given by weak orders rather than numerical
values. For symmetric single-peaked preferences,  
we can circumvent this difficulty by defining an agent's
disutility as the difference between the fraction of red agents in her 
coalition and her ideal ratio, so, given a partition $\pi$, 
for each $i\in N$ we set 
$$
\omega(i)=1 - \left|\frac{|\pi_i\cap R|}{|\pi_i|}-\theta_i\right|.
$$
For uniform and uniform-peak single-peaked preferences, 
we identify the utility of agent
$i$ in partition $\pi$ with the Borda score of 
$\theta(\pi_i) = \frac{|\pi_i\cap R|}{|\pi_i|}$ in $\succ'_i$: 
for each $\theta\in \Theta$, we set 
$$
\beta(i, \theta)=|\{\theta'\in \Theta: \theta\succ'_i \theta'\}|\ \text{ and }
\omega(i)=\frac{\beta(i, \theta(\pi_i))}{|\Theta|-1}.
$$ 
In both cases, we define the average welfare as
$$
\omega(\pi)=\frac{1}{n}\sum_{i\in N}\omega(i).
$$
To gain additional insight into the behavior of our algorithms, 
we also consider two other measures, namely, 
the average coalition size and the average diversity.
We measure the diversity of a coalition $C$
as $\delta(C)=1-2\left|\frac{1}{2}-\frac{|C\cap R|}{|C|}\right|$;
note that $\delta(C)=0$ if $C$ is homogeneous 
and $\delta(C)=1$ if $C$ is balanced.
Both for the average size and for the average diversity, 
we average over all agents in the partition
rather than all coalitions, as we focus on the experience of an individual agent.
That is, the average coalition size and the average diversity are defined, respectively, 
as 
$$
\mu(\pi)=\frac{1}{n}\sum_{i\in N}|\pi_i| \ \ \text{ and }\ \ 
\delta(\pi)=\frac{1}{n}\sum_{i\in N}\delta(\pi_i).
$$
\begin{figure}[t]
\centering
\includegraphics[scale=0.6]{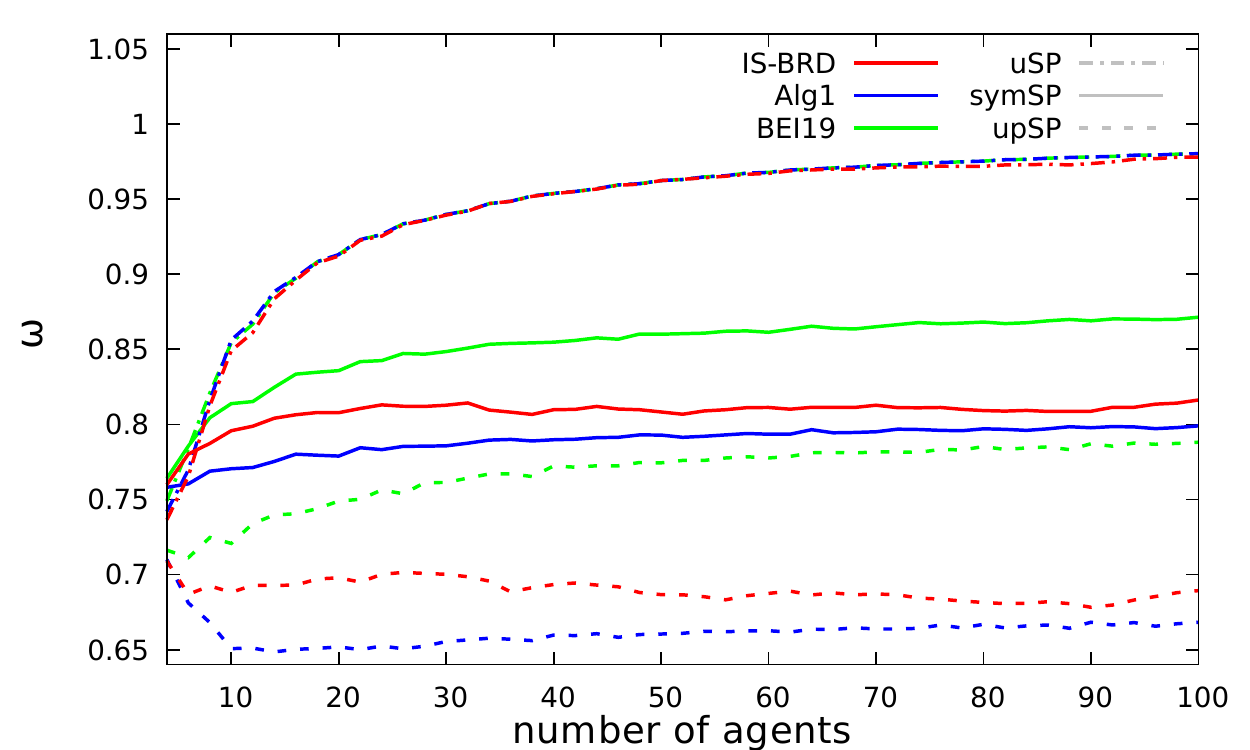}
\caption{Average social welfare} 
\label{fig:sw}
\end{figure}
\begin{figure}[t]
\centering
\includegraphics[scale=0.6]{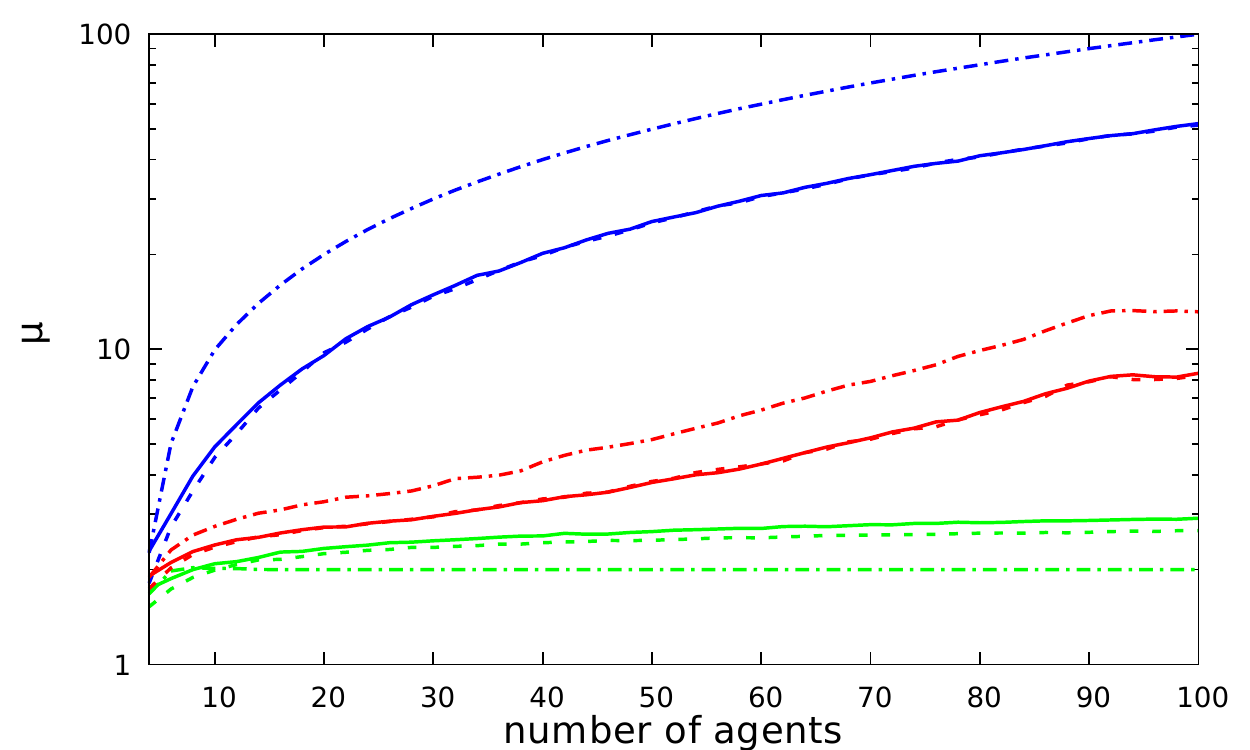}
\caption{Average coalition size} 
\label{fig:size}
\end{figure}
\begin{figure}[t]
\centering
\includegraphics[scale=0.6]{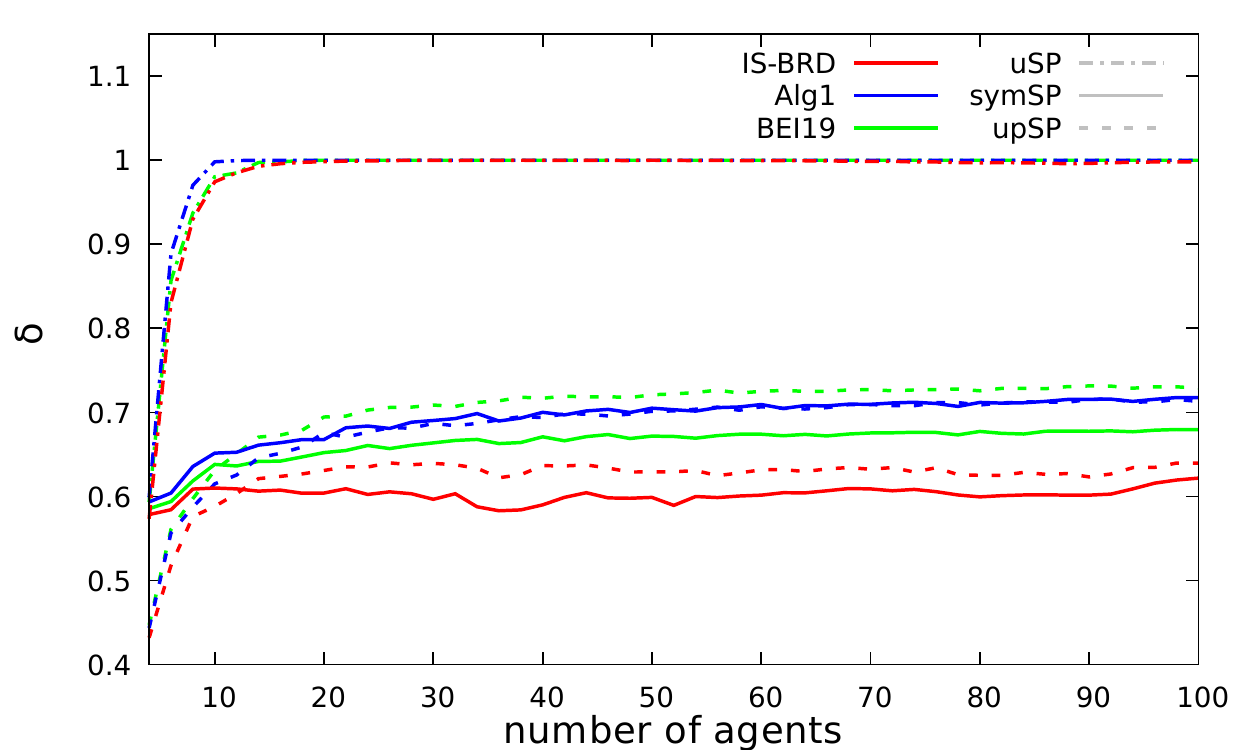}
\caption{Average diversity} 
\label{fig:mix}
\end{figure}

\smallskip

\noindent{\bf Results\ }
The results of our experiments are shown in Figures~\ref{fig:sw}--\ref{fig:mix}.
In all graphs, 
BEI19 refers to the algorithm of Bredereck et al., Alg1 refers to our 
Algorithm~\ref{alg:IS} and IS-BRD refers to 
the IS-BRD algorithm that starts with a uniformly random partition of agents 
and at each iteration chooses the deviation uniformly at random from among all 
available deviations. Remarkably, on every instance we have generated, 
IS-BRD converges in $O(n^2)$ iterations.

We first consider HDGs with the same number
of red and blue agents. For each $s=2, \dots, 50$, we generate 
1000 HDGs with $s$ red and $s$ blue agents; thus, $n=2s$ takes even values 
from 4 to 100.

Our results for social welfare are shown in Figure~\ref{fig:sw}. 
For uSP, all algorithms perform similarly and provide 
fairly high social welfare. Intuitively, this is because under this preference 
model the agents are likely to rank the ratio $\frac12$ highly, and therefore
they are quite happy to be in a coalition with an approximately equal number of 
red and blue agents. Fortunately, all of our algorithms are good at partitioning the agents
into such coalitions. In particular, under this distribution 
it is likely that each agent weakly prefers a balanced coalition to a homogeneous coalition, 
so Algorithm~\ref{alg:IS} will be able to stop after setting $C_0=N$,
and the grand coalition will be well-liked by most agents.
However, the other two distributions tell a different story: under these distributions 
the algorithm of Bredereck et al.~substantially outperforms the other two algorithms, 
with IS-BRD being consistently better than Algorithm~\ref{alg:IS}. Thus, if the distribution
of agents' top choices is close to uniform and the numbers of red and blue agents
are equal, the more complex algorithm of Bredereck et 
al.~has an advantage over our approach.

Figure~\ref{fig:size} shows that the three algorithms are very different
in terms of the sizes of coalitions they produce: Algorithm~\ref{alg:IS}
tends to produce the largest coalitions, and the algorithm of Bredereck et al.~
tends to produce the smallest coalitions; this holds irrespective of how
we sample the agents' preferences.

Figure~\ref{fig:mix} shows that all three algorithms achieve almost perfect diversity
under the uSP distribution; we have suggested possible reasons for this 
when discussing the social welfare. For the other two distributions, 
Algorithm~\ref{alg:IS} and the algorithm of Bredereck et al.~perform similarly,
while IS-BRD results in somewhat less diverse partitions.
  
We also investigate what happens if the two agent classes have different sizes.
\begin{figure}[h]
	\centering
	\includegraphics[scale=0.6]{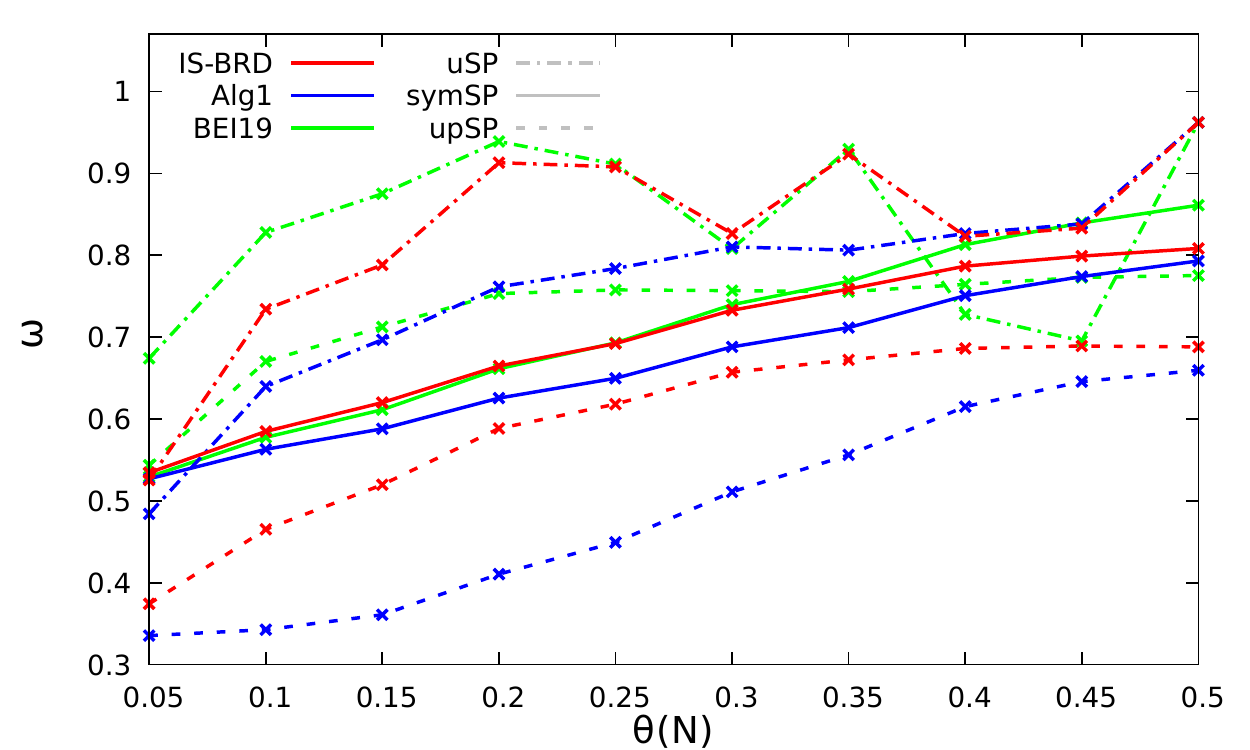}
	\caption{Average social welfare} 
	\label{fig:sw2}
\end{figure}
In Figure~\ref{fig:sw2}, we graph the average social welfare of the individually 
stable outcomes obtained
by the three algorithms we consider, as the number of agents is fixed to $n=50$
and the fraction of red agents in the game (denoted by $\theta(N)$) 
changes from $0$ to $0.5$. 
In this setting, too, the algorithm of Bredereck et al.~outperforms
the other two algorithms with respect to the social welfare, for almost all values
of $\theta(N)=\frac{|R|}{|R|+|B|}$; the difference is the most pronounced
for the upSP distribution. However, for the uSP distribution, there is a range
of values of $\theta(N)$ where Algorithm~\ref{alg:IS} has better performance than
the algorithm of Bredereck et al.; this can be attributed to the fact that 
the latter algorithm only produces coalitions that are very unbalanced or perfectly balanced, 
while Algorithm~1 can output coalitions with more complex ratios. 

\section{Conclusions}\label{sec:concl}
Our work contributes to the study of hedonic diversity games---an 
interesting class of coalition formation games introduced 
by Bredereck et al. 
We have focused on stability concepts that deal with 
deviations by individual agents, namely, Nash stability and individual stability.
%
Remarkably, while our algorithm
for finding IS outcomes is theoretically more appealing 
than the algorithm of Bredereck et al., in that it is simpler and more general, 
empirically the latter algorithm has better performance
for the domain on which it is defined. Perhaps the most interesting 
open question concerning individual stability is whether IS-BRD
is guaranteed to converge to an IS outcome, and if yes, whether convergence
always happens after polynomially many iterations; we note that, 
in our experiments, this is always the case. The same question
can be asked in the context of anonymous games.
We also feel that the $k$-HDG model deserves further attention:
in particular, we would like to understand the complexity 
of finding IS outcomes for small values of $k$ and/or for
single-peaked preferences.

\smallskip

\noindent{\bf Acknowledgements\ }
This work was supported by 
a DFG project ``MaMu'', NI 369/19 (Boehmer) and by
an ERC Starting Grant ACCORD under Grant Agreement 639945 (Elkind).
\bibliographystyle{aaai}
\bibliography{abb,hedonic-refs}

\begin{thebibliography}{}

\bibitem[\protect\citeauthoryear{Aziz and Savani}{2016}]{Aziz2016}
Aziz, H., and Savani, R.
\newblock 2016.
\newblock Hedonic games.
\newblock In Brandt, F.; Conitzer, V.; Endriss, U.; Lang, J.; and Procaccia,
  A., eds., {\em Handbook of Computational Social Choice}. Cambridge University
  Press.
\newblock chapter~15.

\bibitem[\protect\citeauthoryear{Aziz \bgroup et al\mbox.\egroup
  }{2019}]{AzizBBHOP17}
Aziz, H.; Brandl, F.; Brandt, F.; Harrenstein, P.; Olsen, M.; and Peters, D.
\newblock 2019.
\newblock Fractional hedonic games.
\newblock {\em ACM Transactions on Economics and Computation} 7(2):6:1--6:29.

\bibitem[\protect\citeauthoryear{Ballester}{2004}]{Ballester2004}
Ballester, C.
\newblock 2004.
\newblock {N}{P}-competeness in hedonic games.
\newblock {\em Games and Economic Behavior} 49:1--30.

\bibitem[\protect\citeauthoryear{Bogomolnaia and
  Jackson}{2002}]{Bogomolnaia2002}
Bogomolnaia, A., and Jackson, M.
\newblock 2002.
\newblock The stability of hedonic coalition structures.
\newblock {\em Games and Economic Behavior} 38(2):201--230.

\bibitem[\protect\citeauthoryear{Bredereck, Elkind, and Igarashi}{2019}]{BEI19}
Bredereck, R.; Elkind, E.; and Igarashi, A.
\newblock 2019.
\newblock Hedonic diversity games.
\newblock In {\em Proceedings of the 18th International Conference on
  Autonomous Agents and Multi-Agent Systems (AAMAS)},  565--573.

\bibitem[\protect\citeauthoryear{Conitzer}{2009}]{ConitzerSP}
Conitzer, V.
\newblock 2009.
\newblock Eliciting single-peaked preferences using comparison queries.
\newblock {\em Journal of Artificial Intelligence Research} 35:161--191.

\bibitem[\protect\citeauthoryear{Dr\`eze and Greenberg}{1980}]{Dreze1980}
Dr\`eze, J.~H., and Greenberg, J.
\newblock 1980.
\newblock Hedonic coalitions: Optimality and stability.
\newblock {\em Econometrica} 48(4):987--1003.

\bibitem[\protect\citeauthoryear{Echzell \bgroup et al\mbox.\egroup
  }{2019}]{echzell}
Echzell, H.; Friedrich, T.; Lenzner, P.; Molitor, L.; Pappik, M.; Sch{\"{o}}ne,
  F.; Sommer, F.; and Stangl, D.
\newblock 2019.
\newblock Convergence and hardness of strategic {S}chelling segregation.
\newblock In {\em Proceedings of the 16th Conference on Web and Internet
  Economics (WINE)}.

\bibitem[\protect\citeauthoryear{Elkind \bgroup et al\mbox.\egroup
  }{2019}]{jump}
Elkind, E.; Gan, J.; Igarashi, A.; Suksompong, W.; and Voudouris, A.~A.
\newblock 2019.
\newblock Schelling games on graphs.
\newblock In {\em Proceedings of the 28th International Joint Conference on
  Artificial Intelligence ({IJCAI})},  266--272.

\bibitem[\protect\citeauthoryear{Gairing and Savani}{2010}]{Gairing2010}
Gairing, M., and Savani, R.
\newblock 2010.
\newblock Computing stable outcomes in hedonic games.
\newblock In {\em Proceedings of the 3rd International Symposium on Algorithmic
  Game Theory (SAGT)},  174--185.

\bibitem[\protect\citeauthoryear{Garey and Johnson}{1979}]{gj}
Garey, M., and Johnson, D.
\newblock 1979.
\newblock {\em Computers and Intractability: {A} Guide to the Theory of
  {NP}-Completeness}.
\newblock {W. H. Freeman and Company}.

\bibitem[\protect\citeauthoryear{Lackner and Lackner}{2017}]{LL17}
Lackner, M.-L., and Lackner, M.
\newblock 2017.
\newblock On the likelihood of single-peaked preferences.
\newblock {\em Social Choice and Welfare} 48(4):717--745.

\bibitem[\protect\citeauthoryear{Peters}{2016}]{peters2016dichotomous}
Peters, D.
\newblock 2016.
\newblock Complexity of hedonic games with dichotomous preferences.
\newblock In {\em Proceedings of the 30th AAAI Conference on Artificial
  Intelligence (AAAI)},  579--585.

\bibitem[\protect\citeauthoryear{Walsh}{2015}]{Walsh}
Walsh, T.
\newblock 2015.
\newblock Generating single peaked votes.
\newblock {\em CoRR} abs/1503.02766.

\end{thebibliography}
\end{document}